\documentclass[a4paper,12pt]{elsarticle}
\usepackage[T1]{fontenc}
\usepackage[english]{babel}
\usepackage{ae,aecompl}
\usepackage{amsmath, amsthm}
\usepackage{amssymb}
\usepackage[utf8]{inputenc}
\usepackage[section] {placeins}
\usepackage[ruled, vlined, linesnumbered]{algorithm2e}
\newtheorem {Theorem}                 {Theorem}         [section]
\newtheorem {theorem}      [Theorem]  {Theorem}
\newtheorem {myalgorithm}    [Theorem]  {Algorithm}

\newtheorem {Lemma}        [Theorem]  {Lemma}
\newtheorem {lemma}        [Theorem]  {Lemma}

\input amssym.def
\input amssym

\usepackage{pgfplots}
\usepackage{verbatim}
\usepackage{subfigure}
\usepackage{tikz}
\usetikzlibrary{shapes,arrows,backgrounds,mindmap}
\usepackage{titlesec}
\journal{arXiv}

\begin{document}
\begin{frontmatter}
\title{Twinless articulation points and some related problems}
\author{Raed Jaberi}

\begin{abstract}  
	Let $G=(V,E)$ be a twinless strongly connected graph. a vertex $v\in V$ is a twinless articulation point if the subrgraph obtained from $G$ by removing the vertex $v$ is not twinless strongly connected. An edge $e\in E$ is a twinless bridge if the subgraph obtained from $G$ by deleting $e$ is not twiless strongly connected graph. In this paper we study twinless articulation points and twinless bridges.
	We also study the problem of finding a minimum cardinality edge subset $E_{1} \subseteq E$ such that the subgraph $(V,E_{1})$ is twinless strongly connected. Moreover, we present an algorithm for computing the $2$-vertex-twinless connected components of $G$.
\end{abstract} 
\end{frontmatter}
\section{Introduction}
Let $G=(V,E)$ be a twinless strongly connected graph. a vertex $v\in V$ is a twinless articulation point if the subrgraph obtained from $G$ by removing the vertex $v$ is not twinless strongly connected. An edge $e\in E$ is a twinless bridge if the subgraph obtained from $G$ by deleting $e$ is not twiless strongly connected graph.
A twinless strongly connected graph $G$ is $k$-vertex-twinless-connected if $|V|\geq k+1$ and for each $U\subset V$ with $|U|<k$, the induced subgraph on $V\setminus U$ is twinless strongly connected. Thus, a twinless strongly connected graph $G$ is $2$-vertex-twinless-connected if and only if it $|V|\geq 3$ and it does not contain any twinless articulation point. A twinless strongly connected graph $G$ is $2$-edge-twinless-connected if $|V|>2$  and $G$ has no twinless briges. A $2$-vertex-twinless-connected component is a maximal subset $U^{2vt} \subseteq V$ such that the induced subgraph on $U^{2vt}$ is $2$-vertex-twinless-connected (see Figure \ref{fig:newconcepts} and Figure \ref{fig:2vtcgexample}).

 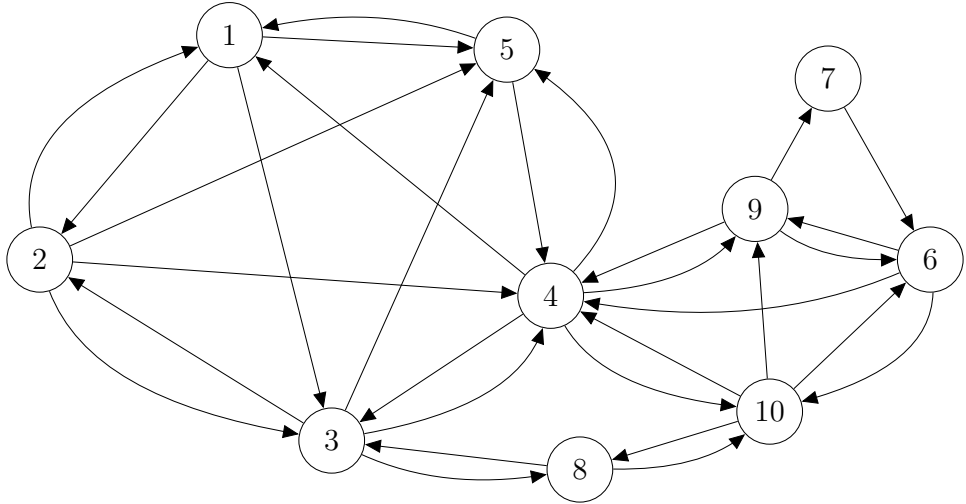
\begin{figure}[htp]
	\centering
	\scalebox{0.96}{
		\begin{tikzpicture}[xscale=2]
		\tikzstyle{every node}=[color=black,draw,circle,minimum size=0.9cm]
		\node (v1) at (-1.2,3.1) {$1$};
		\node (v2) at (-2.5,0) {$2$};
		\node (v3) at (-0.5, -2.5) {$3$};
		\node (v4) at (1,-0.5) {$4$};
		\node (v5) at (0.7,2.9) {$5$};
		\node (v6) at (3.6,0) {$6$};
		\node (v7) at (2.9,2.5) {$7$};
		\node (v8) at (1.2,-2.9) {$8$};
		\node (v9) at (2.4,0.7) {$9$};
		\node (v10) at (2.5,-2.1) {$10$};
		
		\begin{scope}   
		\tikzstyle{every node}=[auto=right]   
		\draw [-triangle 45] (v1) to (v2);
		\draw [-triangle 45] (v2) to [bend left ](v1);
		\draw [-triangle 45] (v1) to (v3);
		\draw [-triangle 45] (v3) to (v5);
		\draw [-triangle 45] (v2) to [bend right ] (v3);
		\draw [-triangle 45] (v3) to (v2);
		\draw [-triangle 45] (v3) to [bend right ](v4);
		\draw [-triangle 45] (v4) to (v3);
		\draw [-triangle 45] (v4) to [bend right ](v5);
		\draw [-triangle 45] (v5) to (v4);
		\draw [-triangle 45] (v1) to (v5);
		\draw [-triangle 45] (v2) to (v5);
		\draw [-triangle 45] (v5) to[bend right ] (v1);
		\draw [-triangle 45] (v4) to(v1);
		\draw [-triangle 45] (v2) to (v4);
		
		\draw [-triangle 45] (v6) to [bend left ](v4);
		\draw [-triangle 45] (v4) to[bend right ] (v9);
		\draw [-triangle 45] (v9) to [bend right ] (v6);
		\draw [-triangle 45] (v6) to (v9);
		\draw [-triangle 45] (v10) to (v6);
		\draw [-triangle 45] (v10) to (v9);
		\draw [-triangle 45] (v8) to (v3);
		\draw [-triangle 45] (v3) to[bend right ] (v8);
		\draw [-triangle 45] (v8) to[bend right ] (v10);
		\draw [-triangle 45] (v10) to (v8);
		\draw [-triangle 45] (v9) to (v7);
		\draw [-triangle 45] (v7) to (v6);
			\draw [-triangle 45] (v10) to (v4);
		\draw [-triangle 45] (v4) to[bend right ]  (v10);
		\draw [-triangle 45] (v6) to[bend left ]  (v10);
		\draw [-triangle 45] (v9) to  (v4);
		\end{scope}
		\end{tikzpicture}}
	\caption{A strongly connected graph $G$, which contains one $2$-vertex-connected component $\lbrace1,2,3,4,5,6,8,9,10\rbrace$, two $2$-vertex-twinless-connected components $\lbrace 1,2,3,4.5\rbrace,\lbrace 9,10,6,4\rbrace$, two articulation points $6,9$, and five twinless articulation points $4,6,9,3,10$. Notice that vertex $4$ is a twinless articulation point but it is not a strong articulation point. Moreover, the $2$-vtcc $\left\lbrace 1,2,3,4,5\right\rbrace $ is a subset of the $2$-vertex-connected component$\lbrace1,2,3,4,5,6,8,9,10\rbrace$. }
	\label{fig:newconcepts}
\end{figure}
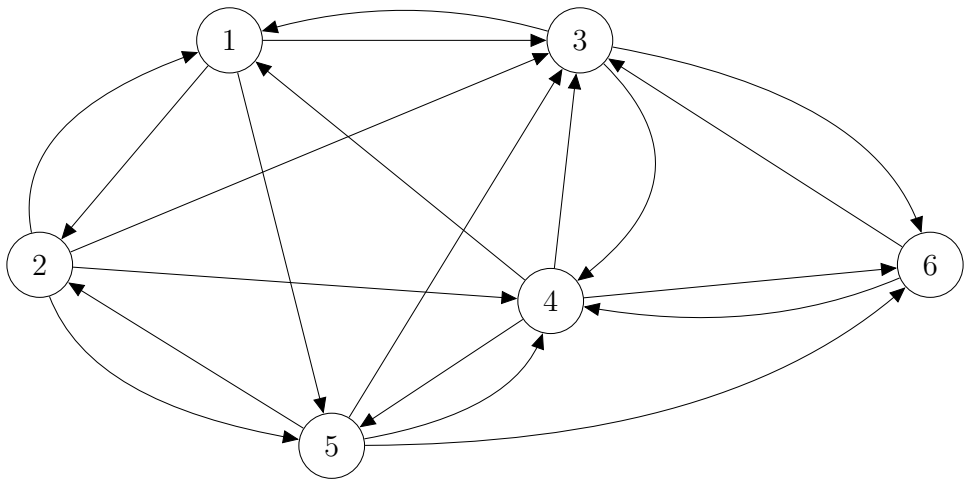
\begin{figure}[htp]
	\centering
	\scalebox{0.96}{
		\begin{tikzpicture}[xscale=2]
		\tikzstyle{every node}=[color=black,draw,circle,minimum size=0.9cm]
		\node (v1) at (-1.2,3.1) {$1$};
		\node (v2) at (-2.5,0) {$2$};
		\node (v5) at (-0.5, -2.5) {$5$};
		\node (v4) at (1,-0.5) {$4$};
		\node (v3) at (1.2,3.1) {$3$};
		\node (v6) at (3.6,0) {$6$};

		\begin{scope}   
		\tikzstyle{every node}=[auto=right]   
		\draw [-triangle 45] (v1) to (v2);
		\draw [-triangle 45] (v2) to [bend left ](v1);
		\draw [-triangle 45] (v1) to (v5);
		\draw [-triangle 45] (v5) to (v3);
		\draw [-triangle 45] (v2) to [bend right ] (v5);
		\draw [-triangle 45] (v5) to (v2);
		\draw [-triangle 45] (v3) to [bend left ](v4);
		\draw [-triangle 45] (v4) to (v5);
		
		\draw [-triangle 45] (v4) to (v3);
		\draw [-triangle 45] (v1) to (v3);
		\draw [-triangle 45] (v2) to (v3);
		\draw [-triangle 45] (v3) to[bend right ] (v1);
		\draw [-triangle 45] (v4) to(v1);
		\draw [-triangle 45] (v2) to (v4);
		
		\draw [-triangle 45] (v6) to [bend left ](v4);
     	\draw [-triangle 45] (v4) to (v6);
	    \draw [-triangle 45] (v6) to (v3);
	    \draw [-triangle 45] (v3) to[bend left ] (v6);
	    \draw [-triangle 45] (v5) to[bend right ] (v4);
	     \draw [-triangle 45] (v5) to[bend right ] (v6);
		\end{scope}
		\end{tikzpicture}}
	\caption{a $2$-vertex-twinless-connected graph $G=(V,E)$. The  graph $G$ is $2$-vertex-connected. Notice that the subgraph  $G\setminus \left\lbrace  (5,6)\right\rbrace $ is $2$-vertex-connected but it is not $2$-vertex-twinless-connected. }
	\label{fig:2vtcgexample}
\end{figure}
Tarjan \cite{T72} gave a linear time algorithm for computing the strongly connected components. In $2006$, Raghavan \cite{SR06} showed that twinless strongly connected component of a directed graph can be calculated in linear time.
In $2010$, Georgiadis \cite{G10} presented an algorithm to check whether a strongly connected graph is $2$-vertex-connected in linear time. Italiano et al. \cite{ILS12} gave linear time algorithms for finding all the strong articulation points and strong bridges of a directed graph. In $2014$, Jaberi \cite{J14} presented algorithms for computing the $2$-vertex-connected components of directed graphs in $O(nm)$ time (published in \cite{Jaberi16}). An experimental study ($2015$) \cite{LGILP15} showed that our  algorithm performs well in practice. Henzinger et al.\cite{HKL15}
gave algorithms for calculating the $2$-vertex-connected components and $2$-edge-connected components of a directed graph in $O(n^{2})$ time. Jaberi \cite{Jaberi15} presented algorithms for computing the $2$-directed blocks, $2$-strong blocks, and $2$-edge blocks of a directed graph. He also presented approximations algorithms for some optimization problems related to these blocks. Georgiadis et al. \cite{GILP14SODA} gave linear time algorithms for finding $2$-edge blocks. The same authors \cite{GILP14VertexConnectivity} gave linear time algorithms for finding $2$-directed blocks and $2$-strong blocks. 

In this paper we study twinless articulation points and twinless bridges.
We also study the following problem, deneoted by MTSCS, of finding a minimum twinless strongly connected spanning subgraph (MTSCS) of a twinless strongly connected graph. It is well known that the problem of finding a minimum strongly connected spanning subgraph (MSCSS) of a strongly connected graph is NP-hard (The Hamiltonian problem can be reduced to this problem \cite{KRY94, ZNI03,CJ79}). Notice that a directed graph $G=(V,E)$ with $|V|\geq 3$ has a Hamiltonian cycle if and only if it conatins a twinless strongly connected subgraph $G_{1}=(V,E_{1})$ with $|E_{1}|\leq n$. Therefore, the MTSCS Problem is also NP-hard.
 Moreover, we present an algorithm for computing the $2$-vertex-twinless connected components of a twinless strongly connected graph.
\section{Graph Terminology and Notation} \label{def:gtan}
A directed graph $G=(V,E)$ is called twinless strongly connected if for each pair of vertices $v,w \in V$ there is a path $p$ from $v$ to $w$ and a path $p_{1}$ from $w$ to $v$ in $G$ such that for each edge $(x,y)$ of $p$, the path $p_{1}$ does not use $(y,x)$. 
A directed graph $G=(V,E)$ is twinless strongly connected if and only if it contains a strongly connected spanning subgraph $(V,E_{1})$ that does not contain antiparallel edges.
Let $G=(V,E)$ be a strongly connected graph. The TSCC component graph $G_{tscc}=(V_{tscc},E_{tscc})$ of $G$ is the directed  graph obtained from $G$ by contracting each twinless strongly connected component into a single supervertex and replacing parallel edges by a single edge. We denote by $E^{r}$ the edge set $\left\lbrace (v,w)\mid (w,v) \in E \right\rbrace $. We use $s_{tap}$ to denote the number of strong articulation points in the directed graph $G$. We denote by $s_{sb}$ the number of strong bridges in $G$.

\section{Computing Twinless Articulation Points}
In this section we illustrate how to compute the twinless articulation points of a strongly connected graph. 
Each strong articulation point is a twinless articulation point, but the converse is not necessarily true. In Figure \ref{fig:newconcepts}, vertex $4$ is a twinless articulation point but it is not a strong articulation point.
Let $G=(V,E)$ be a strongly connected graph and let $w\in V$.
Italiano et al. \cite{ILS12} gave a linear time algorithm for finding strong articulation points in $G$. The algorithm of Italiano et al. \cite{ILS12} calculates the non-trivial dominators in the flowgraphs $(V,E,w)$ and $(V,E^{r},w)$ and tests whether the directed graph $G\setminus\left\lbrace w\right\rbrace $ is strongly connected. In \cite{SR06}, Raghavan proved that $G$ is twinless-strongly-connected if and only if the underlying undriected graph of $G$ is $2$-edge-connected.
By combining the algorithm of Italiano et al. \cite{ILS12} with  Raghavan's algorithm \cite{SR06} we obtain Algorithm \ref{algo:algorithmtwinlessarticulationps} that can compute all the twinless articulation points of $G$.

\begin{figure}[htbp]
	\begin{myalgorithm}\label{algo:algorithmtwinlessarticulationps}\rm\quad\\[-5ex]
		\begin{tabbing}
			\quad\quad\=\quad\=\quad\=\quad\=\quad\=\quad\=\quad\=\quad\=\quad\=\kill
			\textbf{Input:} A strongly connected graph $G=(V,E)$.\\
			\textbf{Output:} The set of all twinless articulation point $T^{ap}$\\
			{\small 1}\> $B \leftarrow \emptyset$\\
			{\small 2}\> choose a vertex $w \in V$, if $G\setminus  \left\lbrace w\right\rbrace $ is not strongly connect, add $w$ to $B$ \\
			{\small 3}\>\textbf{for} each non-trivial dominator $v$ in the flowgraph $(V,E,w)$ \textbf{do} \\
			{\small 4}\>\>$B \leftarrow B\cup \left\lbrace v\right\rbrace $ \\
			{\small 5}\>\textbf{for} each non-trivial dominator $v$ in the flowgraph $(V,E^{r},w)$ \textbf{do} \\
			{\small 6}\>\>$B \leftarrow B\cup \left\lbrace v\right\rbrace $ \\
			{\small 7}\> $T^{ap} \leftarrow B$.\\ 
			{\small 8}\> \textbf{for} each vertex $x \in V \setminus B$ \textbf{do} \\
			{\small 9}\>\> \textbf{if} the underlying graph of $G \setminus \left\lbrace  x\right\rbrace $ \ is not $2$-edge-connected \textbf{then}\\
			{\small 10}\>\>\> $T^{ap}\leftarrow T^{ap} \cup \left\lbrace x \right\rbrace  $.
		\end{tabbing}
	\end{myalgorithm}
\end{figure}

\begin{lemma}
	Let $G=(V,E)$ be a strongly connected graph and let $y \in V$. The vertex $y$ is a twinless articulation point if and only if $y\in T^{ap}$.
\end{lemma}
\begin{proof}
By [\cite{ILS12}, Theorem $5.2.$], a vertex $s\in 	V$ is a strong articulation point if and only if $s \in B$. Let $u$ be any vertex in $T^{ap}\setminus B$. Then the subgraph $G\setminus\left\lbrace u \right\rbrace $ is strongly connected because the vertex $u$ is not a strong articulation point in the graph $G$. . By [\cite{SR06}, Theorem $2$], $G\setminus\left\lbrace u \right\rbrace$ is twinless strongly connected if and only if the underlying graph of $G\setminus\left\lbrace u \right\rbrace $ is  $2$-edge-connected. 
\end{proof}
\begin{theorem}
	Algorithm \ref{algo:algorithmtwinlessarticulationps} runs in $O((n-s_{tap})m)$ time.
\end{theorem}
\begin{proof}
The algorithm of Italiano et al. \cite{ILS12} is able to calculate all the strong articulation points of $G$ in $O(n+m)$ time since non-trivial dominators of flowgraphs $(V,E,w)$ and $(V,E^{r},w)$ can be found in linear time \cite{BGKRTW00,AHLT99,FILOS12}. The number of iterations of the for-loop in lines $8$--$10$ is $n-s_{tap}$. Line $9$ takes $O(n+m)$ time \cite{T74,JS13}. By [\cite{SR06}, Theorem $2$], Raghavan's algorithm \cite{SR06} is able to test whether or not $G\setminus\left\lbrace u \right\rbrace $ is twinless strongly connected in linear time.	
\end{proof}

\section{Approximation Algorithms for MTSCSS Problem}
In this section we present approximation algorithms for the MTSCSS problem. Let $G=(V,E)$ be a strongly connected graph. It is well known that a strongly connected subgraph $G_{sc}=(V,E_{sc})$ of $G$ with $|E_{sc}|\leq 2n-2$ can be obtained from $G$ by finding an outgoing branching and an incoming branching rooted at the same vertex in $G$ \cite{FJ81,KRY94}. If $G$ is twinless strongly connected, then the subgraph $G_{sc}$ is strongly connected but is not necessarily twinless strongly connected.

The following lemma shows how to reduce the number of twinless strongly connected components in a subgraph of a twinless strongly connected graph by adding a small number of edges.
\begin{lemma}\label{def:lemmareducenumberoftsccs}
	Let $G=(V,E)$ be a twinless strongly connected graph, and 
	let $G^{L} = (V,E^{L} )$ be a strongly connected subgraph of $G$ such 
	that  $G^{L}$ contains $k>1$ twinless strongly connected components. Let $v,w$ be two vertices such that $v,w$ lie in different twinless strongly connected components of 
	$G^{L}$ and $(v,w) \in E\setminus E^{L}$. Then $(V,E^{L} \cup \left\lbrace (v,w) \right\rbrace $ has less than $k$ twinless strongly connected components.
\end{lemma}
\begin{proof}
	Since the subgraph $G^{L}$ has at least two twinless strongly connected components, there are two distinct twinless strongly connected components $C_{1},C_{2}$ of $G^{L}$ such that $v \in C_{1} ,w \in C_{2}$ and $(v,w) \in E\setminus E^{L}$ (see Figure \ref{fig:exampleaddingedgetosubgraph}	). Let $G^{L}_{tscc}=(V^{L}_{tscc}.E^{L}_{tscc})$ be the TSCC component graph of the subgraph $G^{L}$. By [\cite{SR06}, Theorem $1$], the underlying graph of $G^{L}_{tscc}$ is a tree and each edge in this tree corresponds to antiparallel edges in $G^{L}$. Therefore, there is a simple path $p$ from $w$ to $v$ in $G^{L}$ such that neither $(v,w)$ nor $(w,v)$ belongs to $p$. The vertices $v$ and $w$ lie in the same twinless strongly connected component in the subgraph $(V,E^{L} \cup \left\lbrace (v,w) \right\rbrace $ because the edge $(w,v)$ does not lie on $p$. By [\cite{SR06}, Lemma $1$], the vertices of $C_{1} \cup C_{2}$ lie in the same twinless strongly connected component of the subgraph $(V,E^{L} \cup \left\lbrace (v,w) \right\rbrace $.
\end{proof}
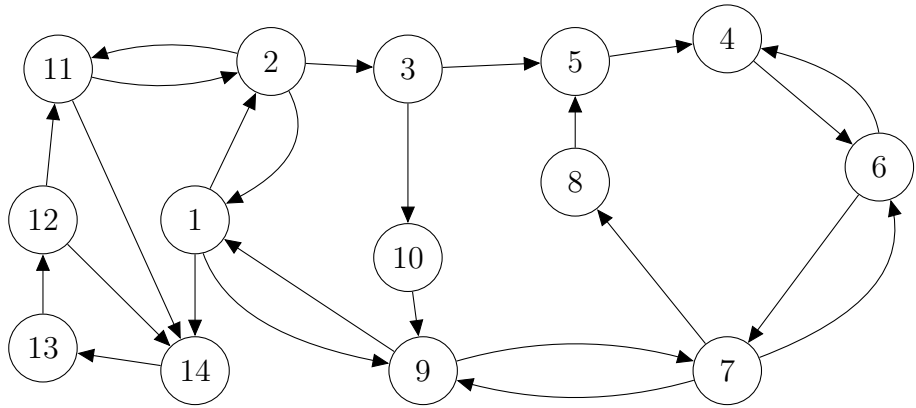
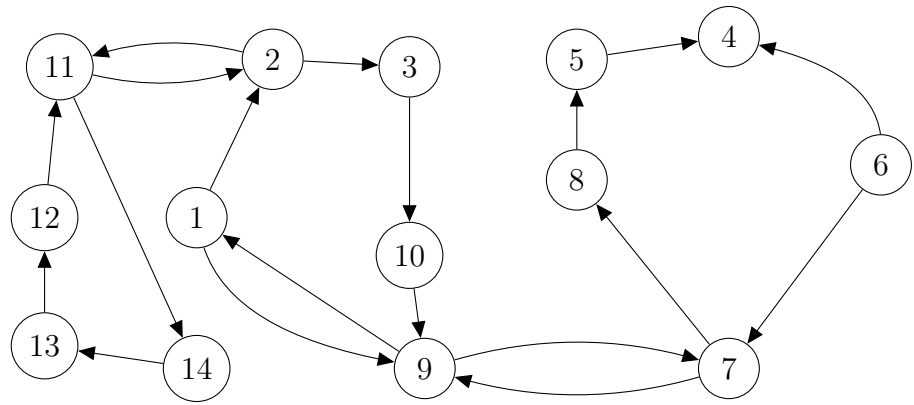
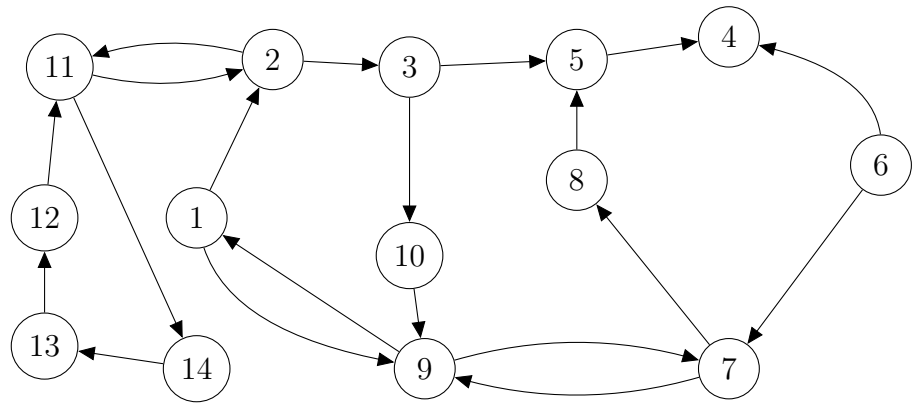
\begin{figure}[htbp]
	\centering
	\subfigure[]{
		\begin{tikzpicture}[xscale=2]
		\tikzstyle{every node}=[color=black,draw,circle,minimum size=0.9cm]
		\node (v1) at (-1.5,0) {$1$};
		\node (v2) at (-1,2.1) {$2$};
		\node (v3) at (-0.1, 2){$3$};
		\node (v4) at (2,2.4) {$4$};
		\node (v5) at (1,2.1) {$5$};
		\node (v6) at (3,0.7) {$6$};
		\node (v7) at (2,-2) {$7$};
		\node (v8) at (1,0.5) {$8$};
		\node (v9) at (0,-2) {$9$};
		\node (v10) at (-0.1,-0.5) {$10$};
		\node (v11) at (-2.4,2) {$11$};
		\node (v12) at (-2.5,0) {$12$};
		\node (v13) at (-2.5,-1.7) {$13$};
		\node (v14) at (-1.5,-2) {$14$};
		
		\begin{scope}   
		\tikzstyle{every node}=[auto=right]   
		\draw [-triangle 45] (v1) to (v2);
		\draw [-triangle 45] (v2) to (v3);
		\draw [-triangle 45] (v3) to (v10);
		\draw [-triangle 45] (v10) to (v9);
		
		\draw [-triangle 45] (v9) to (v1);
		\draw [-triangle 45] (v5) to (v4);
		\draw [-triangle 45] (v6) to (v7);
		\draw [-triangle 45] (v7) to (v8);
		\draw [-triangle 45] (v7) to [bend left](v9);
		\draw [-triangle 45] (v9) to[bend left ] (v7);
		\draw [-triangle 45] (v3) to (v5);
		\draw [-triangle 45] (v8) to (v5);
		\draw [-triangle 45] (v4) to (v6);
		\draw [-triangle 45] (v1) to [bend right](v9);
		\draw [-triangle 45] (v6) to [bend right](v4);
		\draw [-triangle 45] (v7) to [bend right](v6);
		\draw [-triangle 45] (v1) to (v14);
		\draw [-triangle 45] (v14) to (v13);
		\draw [-triangle 45] (v13) to (v12);
		\draw [-triangle 45] (v12) to (v11);
		\draw [-triangle 45] (v11) to[bend right ] (v2);
		\draw [-triangle 45] (v2) to[bend left ] (v1);
		\draw [-triangle 45] (v2) to[bend right ](v11);
		\draw [-triangle 45] (v11) to (v14);
		\draw [-triangle 45] (v12) to (v14);
		\end{scope}
		
		\end{tikzpicture}}
	\subfigure[]{
		\begin{tikzpicture}[xscale=2]
		\tikzstyle{every node}=[color=black,draw,circle,minimum size=0.8cm]
	\node (v1) at (-1.5,0) {$1$};
	\node (v2) at (-1,2.1) {$2$};
	\node (v3) at (-0.1, 2){$3$};
	\node (v4) at (2,2.4) {$4$};
	\node (v5) at (1,2.1) {$5$};
	\node (v6) at (3,0.7) {$6$};
	\node (v7) at (2,-2) {$7$};
	\node (v8) at (1,0.5) {$8$};
	\node (v9) at (0,-2) {$9$};
	\node (v10) at (-0.1,-0.5) {$10$};
	\node (v11) at (-2.4,2) {$11$};
	\node (v12) at (-2.5,0) {$12$};
	\node (v13) at (-2.5,-1.7) {$13$};
	\node (v14) at (-1.5,-2) {$14$};
	
	\begin{scope}   
	\tikzstyle{every node}=[auto=right]   
	\draw [-triangle 45] (v1) to (v2);
	\draw [-triangle 45] (v2) to (v3);
	\draw [-triangle 45] (v3) to (v10);
	\draw [-triangle 45] (v10) to (v9);
	
	\draw [-triangle 45] (v9) to (v1);
	\draw [-triangle 45] (v5) to (v4);
	\draw [-triangle 45] (v6) to (v7);
	\draw [-triangle 45] (v7) to (v8);
	\draw [-triangle 45] (v7) to [bend left](v9);
	\draw [-triangle 45] (v9) to[bend left ] (v7);
	\draw [-triangle 45] (v8) to (v5);
	\draw [-triangle 45] (v1) to [bend right](v9);
	\draw [-triangle 45] (v6) to [bend right](v4);
	\draw [-triangle 45] (v14) to (v13);
	\draw [-triangle 45] (v13) to (v12);
	\draw [-triangle 45] (v12) to (v11);
	\draw [-triangle 45] (v11) to[bend right ] (v2);
	\draw [-triangle 45] (v2) to[bend right ](v11);
	\draw [-triangle 45] (v11) to (v14);
		\end{scope}
		\end{tikzpicture}}

	\subfigure[]{
	\begin{tikzpicture}[xscale=2]
	\tikzstyle{every node}=[color=black,draw,circle,minimum size=0.8cm]
	\node (v1) at (-1.5,0) {$1$};
	\node (v2) at (-1,2.1) {$2$};
	\node (v3) at (-0.1, 2){$3$};
	\node (v4) at (2,2.4) {$4$};
	\node (v5) at (1,2.1) {$5$};
	\node (v6) at (3,0.7) {$6$};
	\node (v7) at (2,-2) {$7$};
	\node (v8) at (1,0.5) {$8$};
	\node (v9) at (0,-2) {$9$};
	\node (v10) at (-0.1,-0.5) {$10$};
	\node (v11) at (-2.4,2) {$11$};
	\node (v12) at (-2.5,0) {$12$};
	\node (v13) at (-2.5,-1.7) {$13$};
	\node (v14) at (-1.5,-2) {$14$};
	
	\begin{scope}   
	\tikzstyle{every node}=[auto=right]   
	\draw [-triangle 45] (v1) to (v2);
	\draw [-triangle 45] (v2) to (v3);
	\draw [-triangle 45] (v3) to (v10);
	\draw [-triangle 45] (v10) to (v9);
	
	\draw [-triangle 45] (v9) to (v1);
	\draw [-triangle 45] (v5) to (v4);
	\draw [-triangle 45] (v6) to (v7);
	\draw [-triangle 45] (v7) to (v8);
	\draw [-triangle 45] (v7) to [bend left](v9);
	\draw [-triangle 45] (v9) to[bend left ] (v7);
		\draw [-triangle 45] (v3) to (v5);
	\draw [-triangle 45] (v8) to (v5);
	\draw [-triangle 45] (v1) to [bend right](v9);
	\draw [-triangle 45] (v6) to [bend right](v4);
	\draw [-triangle 45] (v14) to (v13);
	\draw [-triangle 45] (v13) to (v12);
	\draw [-triangle 45] (v12) to (v11);
	\draw [-triangle 45] (v11) to[bend right ] (v2);
	\draw [-triangle 45] (v2) to[bend right ](v11);
	\draw [-triangle 45] (v11) to (v14);
	\end{scope}
	\end{tikzpicture}}
	\caption{(a) A twinless strongly connected graph $G=(V,E)$. (b) a strongly connected subgraph $G^{L}$ contains $3$ twinless strongly connected components. (c) $(V,E^{L}\cup \left\lbrace (3,5)\right\rbrace )$ has only two twinless strongly connected components}
\label{fig:exampleaddingedgetosubgraph}	
\end{figure}

 \begin{figure}[htbp]
 	\begin{myalgorithm}\label{algo:approximationMTSCSS}\rm\quad\\[-5ex]
 		\begin{tabbing}
 			\quad\quad\=\quad\=\quad\=\quad\=\quad\=\quad\=\quad\=\quad\=\quad\=\kill
 			\textbf{Input:} A twinless strongly connected graph $G=(V,E)$.\\
 			\textbf{Output:} a twinless strongly connected subgraph $G_{t}=(V,E_{t})$\\
 			{\small 1}\> $E_{t} \leftarrow \emptyset$\\
 			{\small 2}\> choose a vertex $w \in V $\\
 			{\small 3}\> construct a spanning tree $T$ rooted at $w$ in $G$\\
 			{\small 4}\>\textbf{for} each edge $e$ in $T$ \textbf{do} \\
 			{\small 5}\>\>$E_{t} \leftarrow E_{t}\cup \left\lbrace e\right\rbrace $ \\
 			{\small 6}\> construct a spanning tree $T_{r}$ rooted at $w$ in $G^{r}=(V,E^{r})$\\
 			{\small 7}\>\textbf{for} each edge $(i,j)$ in $T$ \textbf{do} \\
 			{\small 8}\>\>$E_{t} \leftarrow E_{t}\cup \left\lbrace (j,i)\right\rbrace $ \\
 			{\small 9}\> \textbf{while} $G_{t}=(V,E_{t})$ is not twinless strongly connected \textbf{do}\\ 
 			{\small 10}\>\> find the twinless strongly connected components of $G_{t}$\\
 			{\small 11}\>\> find an edge $(v,w) \in E\setminus E_{t}$ such that $v,w $ are in distinct\\
 			 {\small 12}\>\>\>twinless strongly connected components of $G_{t}$\\
 			{\small 13}\>\>  $E_{t}\leftarrow E_{t} \cup \left\lbrace (v,w) \right\rbrace  $.\\
 		\end{tabbing}
 	\end{myalgorithm}
 \end{figure}

\begin{Theorem}
	Algorithm \ref{algo:approximationMTSCSS} has an approximation factor of $3$. 
\end{Theorem}
\begin{proof}
	Clearly, each optimal solution for the MTSCSS problem has at least $n$ edge. Moreover, $T \cup T_{r}$ contains at most $2n-2$ edges. In each iteration of the while loop (lines $9$--$13$), we add only an edge to $E_{t}$. By Lemma \ref{def:lemmareducenumberoftsccs}, $E_{t}$ contains at most $3n-3$ edges.  
\end{proof}
The correctness of Algorithm \ref{algo:approximationMTSCSS} follows from Lemma \ref{def:lemmareducenumberoftsccs}.

\begin{Theorem}\label{def:runningtimeapproximationMTSCSS}
	The running time of Algorithm \ref{algo:approximationMTSCSS} is $O(nm)$.
\end{Theorem}
\begin{proof}
A spanning tree of a twinless strongly connected graph can found in linear time using DFS. By Lemma \ref{def:lemmareducenumberoftsccs}, the number of iterations of the while loop is at most $n-1$. Furthermore, the twinless strongly connected components of $G_{t}$ can be calculated in linear time using Raghavan's algorithm \cite{SR06}.
\end{proof}

Let $G^{L}=(V,E^{L})$ be a strongly connected subgraph of a twinless strongly connected graph $G=(V,E)$ such that $G^{L}$ is not twinless strongly connected.
Let $v,w$ be two distinct vertices in $G^{L}$ such that $(v,w) \in E\setminus E^{L} $ and $v,w$  don't lie in the same twinless strongly connected component of $G^{L}$.  There is a simple path $p$ from $w$ to $v$ in $G^{L}$ because the subgraph $G^{L}$ is strongly connected. We use $S^{wv}_{p}$ to denote the set of edges $\{ (i,j) \mid  (j,i) \in E^{L} $ belongs to the path $p$ and there exists an edge in the underlying graph of TSCC component graph of  $G^{L}$ which corresponds to $(i,j)$ in $ G^{L}\} $. 

The following lemma leads to an approximation algorithm (Algorithm \ref{algo:refinedapproximationMTSCSS}) for the MTSCSS problem.

\begin{lemma}\label{def:approximationalgorithmfortscss }
	Let $G=(V,E)$ be a twinless strongly connected graph, and let  $G^{L}=(V,E^{L})$ be a strongly connected subgraph of $G$ such that $G^{L}$  has $k$ twinless strongly connected components and $v,w \in V$ don't lie in the same twinless strongly connected component of $G^{L}$. Then $|S^{wv}_{p}|>0$ and $(V,(E^{L} \cup \left\lbrace (v,w)\right\rbrace)\setminus S^{wv}_{p})$ contains at most $k-1$ twinless strongly connected components, where $p$ is a simple path from $w$ to $v$ in $G^{L}$.  
\end{lemma} 
\begin{proof}
 Assume that $C_{1},C_{2}$ are two distinct twinless strongly connected components of $G^{L}$ of such that $v\in C_{1}$ and $w\in C_{2}$. By [\cite{SR06}, Theorem $1$], the underlying graph of the TSCC component graph of $G^{L}$ is a tree $T^{L}$ and each edge $(u_{i},u_{j})$ in $T^{L}$ corresponds to a pair of antiparallel edges in $G^{L}$. Suppose that $p_{tscc}=(v_{1},v_{2},\dots,v_{r})$ is a path from $v_{1}$ to $v_{r}$ in $T^{L}$, where $v_{1}$ corresponds to the component $C_{1}$ and $v_{r}$ corresponds to the component $C_{2}$. Each vertex of $p_{tscc}$ corresponds to a twinless strongly connected component of $G^{L}$. Let $(x,y)$ be an edge on $p$ such that $(x,y)$ corresponds to an edge $(v_{x},v_{y})$ on $p_{tscc}$. Then $(y,x) \in S^{wv}_{p}$. The path $p_{tscc}$ has length at least $1$. Therefore, $|S^{wv}_{p}|>0$. Since the path $p$ and the edge $(v,w)$ form a cycle in $G^{L}$, all the twinless strongly cnnected components that correspond to the vertices of $p_{tscc}$ are merged into a single component in $(V,(E^{L} \cup \left\lbrace (v,w)\right\rbrace)\setminus S^{wv}_{p})$. Thus, the subgraph $(V,(E^{L} \cup \left\lbrace (v,w)\right\rbrace)\setminus S^{wv}_{p})$ has at most $k-1$ twinless strongly connected components.
	
\end{proof}

\begin{figure}[h]
	\begin{myalgorithm}\label{algo:refinedapproximationMTSCSS}\rm\quad\\[-5ex]
		\begin{tabbing}
			\quad\quad\=\quad\=\quad\=\quad\=\quad\=\quad\=\quad\=\quad\=\quad\=\kill
			\textbf{Input:} A twinless strongly connected graph $G=(V,E)$.\\
			\textbf{Output:} a twinless strongly connected subgraph $G_{t}=(V,E_{t})$\\
			{\small 1}\> $E_{t} \leftarrow \emptyset$\\
			{\small 2}\> choose a vertex $w \in V $\\
			{\small 3}\> construct a spanning tree $T$ rooted at $w$ in $G$\\
			{\small 4}\>\textbf{for} each edge $e$ in $T$ \textbf{do} \\
			{\small 5}\>\>$E_{t} \leftarrow E_{t}\cup \left\lbrace e\right\rbrace $ \\
			{\small 6}\> construct a spanning tree $T_{r}$ rooted at $w$ in $G^{r}=(V,E^{r})$\\
			{\small 7}\>\textbf{for} each edge $(i,j)$ in $T$ \textbf{do} \\
			{\small 8}\>\>$E_{t} \leftarrow E_{t}\cup \left\lbrace (j,i)\right\rbrace $ \\
			{\small 9}\> \textbf{while} $G_{t}=(V,E_{t})$ is not twinless strongly connected \textbf{do}\\ 
			{\small 10}\>\> calculate the twinless strongly connected components of $G_{t}$\\
			{\small 11}\>\> find an edge $(v,w) \in E\setminus E_{t}$ such that $v,w $ are in distinct\\
			{\small 12}\>\>\>twinless strongly connected components of $G_{t}$\\
			{\small 13}\>\> find a a simple path $p$ from $w$ to $v$ in $G_{t}$\\
			{\small 14}\>\>  $E_{t}\leftarrow E_{t} \cup \left\lbrace (v,w) \right\rbrace  $.\\
			{\small 15}\>\>\textbf{for} each edge $(i,j)$ on $p$ \textbf{do}\\
			{\small 16}\>\>\> Let $T^{L}$ be the TSCC component graph of $G_{t}$.\\ 
			{\small 17}\>\>\> \textbf{if} $(i,j)$ corresponds to an edge in $T^{L}$ \textbf{then}\\
			{\small 18}\>\>\>\>$E_{t}\leftarrow E_{t}\setminus \left\lbrace (j,i) \right\rbrace $

		\end{tabbing}
	\end{myalgorithm}
\end{figure}

.
 \begin{lemma}
 	Algorithm \ref{algo:refinedapproximationMTSCSS} has approximation factor of $2$.
 \end{lemma}
\begin{proof}
	The number of edges of the subgraph computed in lines $1$--$8$ of Algorithm \ref{algo:refinedapproximationMTSCSS} is at most $ 2n-2$. By Lemma \ref{def:approximationalgorithmfortscss }, in each iteration of the while loop, an edge is added in line $14$ but at least one edge is removed in lines $15$--$18$,  

\end{proof}
 \begin{Theorem}
 	Algorithm \ref{algo:refinedapproximationMTSCSS} runs in $O(nm)$ time.
 \end{Theorem}
\begin{proof}
	Lines $14$--$17$ take linear time.
\end{proof}
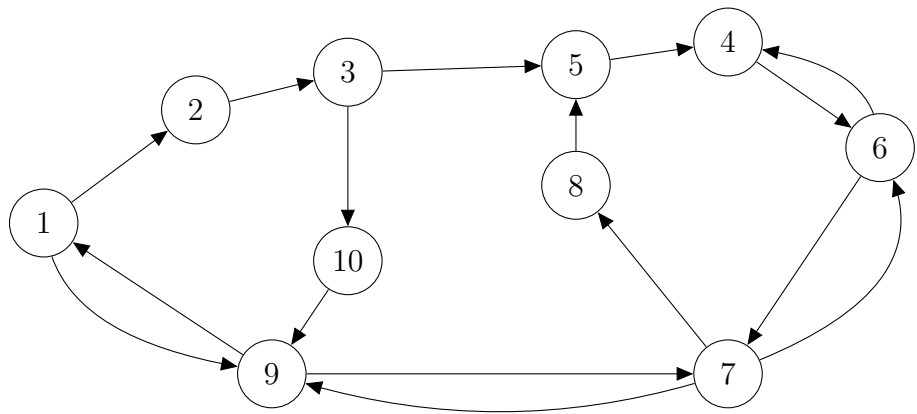
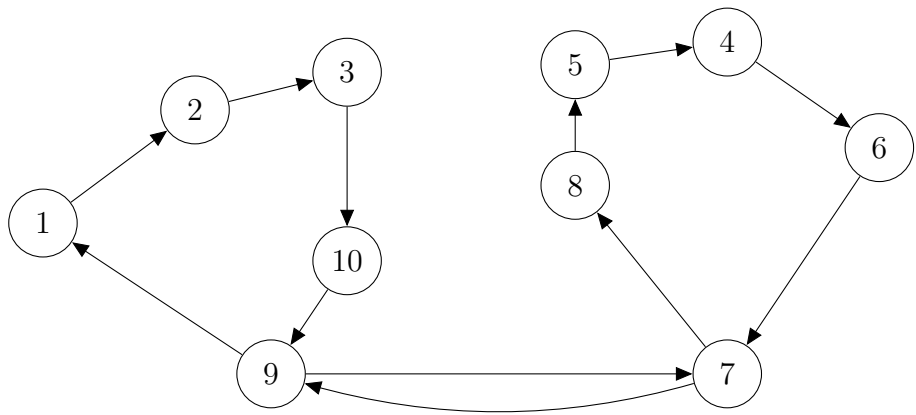
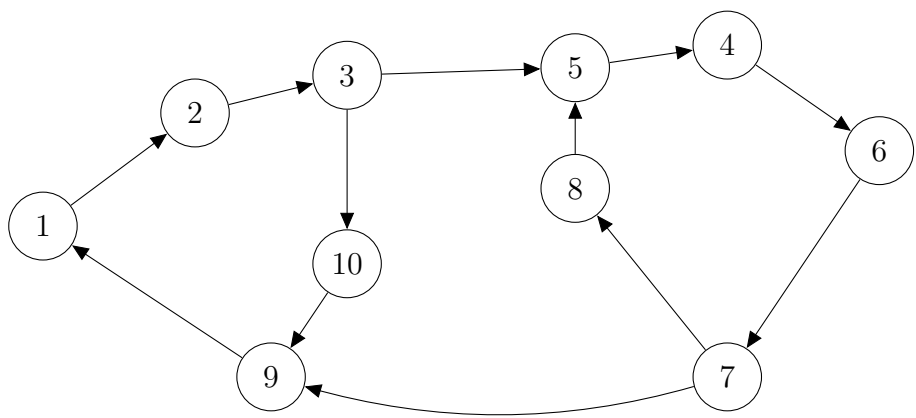
\begin{figure}[htbp]
	\centering
	\subfigure[]{
		\begin{tikzpicture}[xscale=2]
		\tikzstyle{every node}=[color=black,draw,circle,minimum size=0.9cm]
		\node (v1) at (-2.5,0) {$1$};
		\node (v2) at (-1.5,1.5) {$2$};
		\node (v3) at (-0.5, 2){$3$};
		\node (v4) at (2,2.4) {$4$};
		\node (v5) at (1,2.1) {$5$};
		\node (v6) at (3,1) {$6$};
		\node (v7) at (2,-2) {$7$};
		\node (v8) at (1,0.5) {$8$};
		\node (v9) at (-1,-2) {$9$};
		\node (v10) at (-0.5,-0.5) {$10$};
		
		\begin{scope}   
		\tikzstyle{every node}=[auto=right]   
	   	\draw [-triangle 45] (v1) to (v2);
	   \draw [-triangle 45] (v2) to (v3);
	   \draw [-triangle 45] (v3) to (v10);
	   \draw [-triangle 45] (v10) to (v9);
	 
	   \draw [-triangle 45] (v9) to (v1);
	   \draw [-triangle 45] (v5) to (v4);
	   \draw [-triangle 45] (v6) to (v7);
	   \draw [-triangle 45] (v7) to (v8);
	   \draw [-triangle 45] (v7) to [bend left](v9);
	   \draw [-triangle 45] (v9) to (v7);
	   \draw [-triangle 45] (v3) to (v5);
	   \draw [-triangle 45] (v8) to (v5);
	   \draw [-triangle 45] (v4) to (v6);
	   \draw [-triangle 45] (v1) to [bend right](v9);
	   \draw [-triangle 45] (v6) to [bend right](v4);
	   \draw [-triangle 45] (v7) to [bend right](v6);
		\end{scope}
		
		\end{tikzpicture}}
	\subfigure[]{
	\begin{tikzpicture}[xscale=2]
	\tikzstyle{every node}=[color=black,draw,circle,minimum size=0.8cm]
	
	\tikzstyle{every node}=[color=black,draw,circle,minimum size=0.9cm]
\node (v1) at (-2.5,0) {$1$};
\node (v2) at (-1.5,1.5) {$2$};
\node (v3) at (-0.5, 2){$3$};
\node (v4) at (2,2.4) {$4$};
\node (v5) at (1,2.1) {$5$};
\node (v6) at (3,1) {$6$};
\node (v7) at (2,-2) {$7$};
\node (v8) at (1,0.5) {$8$};
\node (v9) at (-1,-2) {$9$};
\node (v10) at (-0.5,-0.5) {$10$};
	
	\begin{scope}   
	\tikzstyle{every node}=[auto=right]   
	\draw [-triangle 45] (v1) to (v2);
	\draw [-triangle 45] (v2) to (v3);
	\draw [-triangle 45] (v3) to (v10);
	\draw [-triangle 45] (v10) to (v9);
	
	\draw [-triangle 45] (v9) to (v1);
	\draw [-triangle 45] (v5) to (v4);
	\draw [-triangle 45] (v6) to (v7);
	\draw [-triangle 45] (v7) to (v8);
	\draw [-triangle 45] (v7) to [bend left](v9);
	\draw [-triangle 45] (v9) to (v7);
	
	\draw [-triangle 45] (v8) to (v5);
	\draw [-triangle 45] (v4) to (v6);
 
	\tikzstyle{every node}=[auto=right] 

	\end{scope}
	\end{tikzpicture}}
	\subfigure[]{
	\begin{tikzpicture}[xscale=2]
	\tikzstyle{every node}=[color=black,draw,circle,minimum size=0.9cm]
	\node (v1) at (-2.5,0) {$1$};
	\node (v2) at (-1.5,1.5) {$2$};
	\node (v3) at (-0.5, 2){$3$};
	\node (v4) at (2,2.4) {$4$};
	\node (v5) at (1,2.1) {$5$};
	\node (v6) at (3,1) {$6$};
	\node (v7) at (2,-2) {$7$};
	\node (v8) at (1,0.5) {$8$};
	\node (v9) at (-1,-2) {$9$};
	\node (v10) at (-0.5,-0.5) {$10$};
	
	\begin{scope}   
	\tikzstyle{every node}=[auto=right]   
	\draw [-triangle 45] (v1) to (v2);
	\draw [-triangle 45] (v2) to (v3);
	\draw [-triangle 45] (v3) to (v10);
	\draw [-triangle 45] (v10) to (v9);
	
	\draw [-triangle 45] (v9) to (v1);
	\draw [-triangle 45] (v5) to (v4);
	\draw [-triangle 45] (v6) to (v7);
	\draw [-triangle 45] (v7) to (v8);
	\draw [-triangle 45] (v7) to [bend left](v9);
	
	\draw [-triangle 45] (v3) to (v5);
	\draw [-triangle 45] (v8) to (v5);

	\draw [-triangle 45] (v4) to (v6);
	\end{scope}
	
	\end{tikzpicture}}
\caption{(a) A twinless strongly connected graph $G=(V,E)$. (b) An optimal solution $E_{s}$ for MSCSS problem with $|E_{s}|=12$. (c) An optimal solution $E_{t}$ for MTSCSS problem with $|E_{t}|=12$.}
\label{fig:exampleoptmtscssmscss}
\end{figure}
\section{Relationship between MTSCSS Problem and MSCSS Problem}
In this section we show a relationship between the MTSCSS problem and the MSCSS problem. 
The following lemma provides a relationship between optimal solutions for the MTSCSS problem and optimal solutions for the MSCSS problem. as explained in Figure \ref{fig:exampleoptmtscssmscss}.
\begin{lemma}\label{def:mscssequalsmtscss}
	 Let $G=(V,E)$ be a twinless strongly connected graph. Let $E_{t}\subseteq E$ be an optimal solution for MTSCSS problem and let $E_{s}\subseteq E$ be an optimal solution for MSCSS problem. Then $|E_{t}|=|E_{s}|$.
\end{lemma}
\begin{proof}
	Since each twinless strongly connected spanning subgraph is strongly connected, $|E_{t}|\geq |E_{s}$|.
	
	Now we prove that $|E_{t}| \leq |E_{s}|$. Suppose that the subgraph $(V,E_{s})$ is not twinless strongly connected. We can use lines $9$--$18$ of Algorithm \ref{algo:refinedapproximationMTSCSS}  to transform it to a twinless strongly connected subgraph $G_{L}=(V,E_{L})$. By lemma \ref{def:approximationalgorithmfortscss }, $|E_{L}|\leq |E_{s}|$. Since $G_{L}$ is twinless strongly connected, we have $|E_{L}|\geq |E_{t}|$.

\end{proof}

 \begin{figure}[htbp]
	\begin{myalgorithm}\label{algo:approximationMTSCSSMSS}\rm\quad\\[-5ex]
		\begin{tabbing}
			\quad\quad\=\quad\=\quad\=\quad\=\quad\=\quad\=\quad\=\quad\=\quad\=\kill
			\textbf{Input:} A twinless strongly connected graph $G=(V,E)$.\\
			\textbf{Output:} a twinless strongly connected subgraph $G_{t}=(V,E_{t})$\\
			{\small 1}\> Let $(V,E_{t})$ be a solution for MSCSS problem computed \\
		    {\small 2}\>\>by approximation algorithm A \\
			{\small 3}\> \textbf{while} $G_{t}=(V,E_{t})$ is not twinless strongly connected \textbf{do}\\ 
		{\small 4}\>\> calculate the twinless strongly connected components of $G_{t}$\\
		{\small 5}\>\> find an edge $(v,w) \in E\setminus E_{t}$ such that $v,w $ are in distinct\\
		{\small 6}\>\>\>twinless strongly connected components of $G_{t}$\\
		{\small 7}\>\> find a a simple path $p$ from $w$ to $v$ in $G_{t}$\\
		{\small 8}\>\>  $E_{t}\leftarrow E_{t} \cup \left\lbrace (v,w) \right\rbrace  $.\\
		{\small 9}\>\>\textbf{for} each edge $(i,j)$ on $p$ \textbf{do}\\
		{\small 10}\>\>\> Let $T^{L}$ be the TSCC component graph of $G_{t}$.\\ 
		{\small 11}\>\>\> \textbf{if} $(i,j)$ corresponds to an edge in $T^{L}$ \textbf{then}\\
		{\small 12}\>\>\>\>$E_{t}\leftarrow E_{t}\setminus \left\lbrace (j,i) \right\rbrace $
		\end{tabbing}
	\end{myalgorithm}
\end{figure}
\begin{Theorem}
	If $A$ has an approximation factor of $\alpha$, then Algorithm \ref{algo:approximationMTSCSSMSS} is an $\alpha$-approximation algorithm for the MTSCSS problem.
\end{Theorem}
\begin{proof}
It follows from Lemma \ref{def:mscssequalsmtscss} and Lemma \ref{def:approximationalgorithmfortscss }.	
\end{proof}
\begin{theorem}
  Algorithm \ref{algo:approximationMTSCSSMSS} runs in  $O(nm+ t_{A})$ time, where $t_{A}$ is the running time of the algorithm A.
\end{theorem}
\begin{proof}
By Lemma \ref{def:approximationalgorithmfortscss }, the number of iterations of the while loop is at most $n-1$, and each iteration takes $O(m)$ time. 
\end{proof}
 Khuller et al. \cite{KRY94} presented approximation algorithms for the MSCSS problem. Furthermore, Zhao et al. \cite{ZNI03} introduced a linear time $5/3$ approximation algorithm for the MSCSS problem.  

Now we show that the running time of Algorithm \ref{algo:approximationMTSCSSMSS} can be improved by using union find data structure. Note that in each iteration of the do-while loop, at least two twinless strongly connected components are merged by adding edge $(v,w)$ to $G_{t}$. Instead of finding twinless strongly connected components in each iteration we compute these components just once then we update them after adding edges. 
 \begin{figure}[htbp]
	\begin{myalgorithm}\label{algo:approximationMTSCSSMSSunionfind}\rm\quad\\[-5ex]
		\begin{tabbing}
			\quad\quad\=\quad\=\quad\=\quad\=\quad\=\quad\=\quad\=\quad\=\quad\=\kill
			\textbf{Input:} A twinless strongly connected graph $G=(V,E)$.\\
			\textbf{Output:} a twinless strongly connected subgraph $G_{t}=(V,E_{t})$\\
			{\small 1}\> Let $(V,E_{t})$ be a solution for MSCSS problem computed \\
			{\small 2}\>\>by approximation algorithm A \\
			{\small 3}\> \textbf{for} every twinless strongly connected components $C$ of $G_{t}$ \textbf{do}\\
			{\small 4}\>\> choose a vertex $w \in C$\\
			{\small 5}\>\> \textbf{for} every vertex $v\in C\setminus \left\lbrace w \right\rbrace $ \textbf{do} \\
			{\small 6}\>\>\> UNION$(v,w)$ \\
			{\small 7}\> Build  the underlying graph $T^{R}$ of TSCC component graph of $G_{t}$.\\	
			{\small 8}\> \textbf{for} each edge $(v,w) \in E\setminus E_{t}$ \textbf{do}\\
			{\small 9}\>\> \textbf{if} Find$(v) \neq$ Find$(w)$ \textbf{then}\\
		
			{\small 10}\>\>\> find a path $p$ from $C_{v}$ to $C_{w}$ in $T^{R}$, where $C_{i}$ is a vertex in\\
	    	{\small 11}\>\>\>  $T^{R}$ corresponding to the twinless strongly connected\\
	    	{\small 12}\>\>\>  component of $G_{t}$ which contains $i$.\\
	    	{\small 13}\>\>\> Contracting the vertices of $p$ into a single vertex in $T^{R}$.\\
	    	{\small 14}\>\>\>$E_{t}\leftarrow E_{t} \cup \left\lbrace (v,w) \right\rbrace  $.\\
	    	{\small 15}\>\>\> \textbf{for} each edge $(x,y)$ on $p$ \textbf{do}\\
	    	{\small 16}\>\>\>\> \textbf{if} $(y,x) \in E_{t} $ \textbf{then}\\
	    	{\small 17}\>\>\>\>\> remove $(y,x)$ from $E_{t}$.\\ 
	    	{\small 18}\>\>UNION$(v,w)$

		\end{tabbing}
	\end{myalgorithm}
\end{figure}

\begin{Theorem}
	Algorithm \ref{algo:approximationMTSCSSMSSunionfind} runs in $(n^{2}+ t_{A})$ time, where $t_{A}$ is the running time of the algorithm A.
\end{Theorem}
\begin{proof}
	By Lemma \ref{def:lemmareducenumberoftsccs} ,the number of UNION operations is at most $n-1$. Since $|E\setminus E_{t}|\leq m$, the number of F
	ind operations is at most $2m$. Furthermore, calculating a path $p$ in $T^{R}$ requires $O(n)$ time and updating $T^{R}$ in line $13$ requires $O(n)$ time. UNION and Find operations can be implemented in almost linear time (see \cite{CLRS09}).
\end{proof}
\section{Computing Twinless bridges}
In this section, we explain how to calculate twinless bridges in strongly connected graph by combining the algorithm of Italiano et al. \cite{ILS12} with  Raghavan's algorithm \cite{SR06}.
\begin{Lemma}\label{def:lemmasbtb}
	Let $G=(V,E)$ be a strongly connnected graph. Then each strong bridge in $G$ is a twinless birdge $G$.
\end{Lemma}
\begin{proof}
	Immediate from the definition.
\end{proof}
Note that the converse of Lemma \ref{def:lemmasbtb} is not necessarily true. Algorithm \ref{algo:algorithmtwinlessbirdges} can find all the twinless bridges. Its correctness follows from Lemma \ref{def:lemmasbtb} and [\cite{SR06}, Theorem $2$].

\begin{figure}[h]
	\begin{myalgorithm}\label{algo:algorithmtwinlessbirdges}\rm\quad\\[-5ex]
		\begin{tabbing}
			\quad\quad\=\quad\=\quad\=\quad\=\quad\=\quad\=\quad\=\quad\=\quad\=\kill
			\textbf{Input:} A strongly connected graph $G=(V,E)$.\\
			\textbf{Output:} The set of all twinless bridges $T^{b}$\\
			{\small 1}\> $A \leftarrow \emptyset$\\
			{\small 2}\> Compute the set of strong bridges in $G$\\
			{\small 3}\>\textbf{for} each strong bridge $e$ in $G$ \textbf{do} \\
			{\small 4}\>\>$A \leftarrow A\cup \left\lbrace e\right\rbrace $ \\
			{\small 5}\> $T^{b} \leftarrow A$.\\ 
			{\small 6}\> \textbf{for} each edge $e\in E \setminus A$ \textbf{do} \\
			{\small 7}\>\> \textbf{if} the underlying graph of $G \setminus \left\lbrace  e\right\rbrace $ \ is not $2$-edge-connected \textbf{then}\\
			{\small 8}\>\>\> $T^{b}\leftarrow T^{b} \cup \left\lbrace e\right\rbrace  $.
		\end{tabbing}
	\end{myalgorithm}
\end{figure}
\begin{Theorem}
	The twinless bridges of a strongly connected graph can be computed in $O((m-s_{tb})m)$
\end{Theorem}
\begin{proof}
	Italiano et al. \cite{ILS12,FILOS12} showed that strong bridges can be calculated in linear time. The number of iterations of the for-loop in lines $6$--$8$ is $m-s_{tb}$. For every edge $e \in E \setminus A$, the subgraph $G\setminus\left\lbrace e \right\rbrace $ is strongly connected. Therefore, by [\cite{SR06}, Theorem $2$], Raghavan's algorithm \cite{SR06} can test whehter or not $G\setminus\left\lbrace e\right\rbrace $ is twinless strongly connected in $O(m)$ time. Checking $2$-edge-connected of an undirected connected graph can be done in $O(m)$ time. Thus, line $7$ takes $O(m)$ time.
\end{proof}
In \cite{ILS12} Italiano et al. proved that the strong bridges of a strongly connected graph $G$ can be found in $O(n(m+n)$ time by calculating a subgraph consisting of inbranching and outbranching rooted at the same vertex and testing each edge of this subgraph whether it is a strong bridge. We can modify this algorithm to find all twinless bridges in a twinless strongly connected graphs.

\begin{figure}[h]
	\begin{myalgorithm}\label{algo:algorithmtwinlessbirdgesusiingsubgraph}\rm\quad\\[-5ex]
		\begin{tabbing}
			\quad\quad\=\quad\=\quad\=\quad\=\quad\=\quad\=\quad\=\quad\=\quad\=\kill
			\textbf{Input:} A twinless strongly connected graph $G=(V,E)$.\\
			\textbf{Output:} The set of all twinless bridges $A$\\
			{\small 1}\> $A \leftarrow \emptyset$\\
			{\small 2}\> Let $B_{s}$ be the set of strong bridges in $G$\\
			{\small 3}\>\textbf{for} each edge $e \in B_{s}$  \textbf{do} \\
			{\small 4}\>\>$A \leftarrow A\cup \left\lbrace e\right\rbrace $ \\
			 {\small 5}\> compute a subgraph $G_{t}=(V,E_{t})$ of $G$ using Algorithm \ref{algo:approximationMTSCSS}\\
			{\small 6}\>\textbf{for} each edge $e\in E_{t}\setminus B_{s}$  \textbf{do} \\
			{\small 7}\>\> \textbf{if} $G\setminus \left\lbrace e \right\rbrace $ is not twinless strongly connected \textbf{then}\\
			{\small 8}\>\>\>$A \leftarrow A\cup \left\lbrace e\right\rbrace $ 
		\end{tabbing}
	\end{myalgorithm}
\end{figure}
\begin{theorem}
	The running time of Algorithm \ref{algo:algorithmtwinlessbirdgesusiingsubgraph} is $O(nm)$.
\end{theorem}
\begin{proof}
	The set $B_{s}$ can be found in linear time using the algorithm of Italiano et al. \cite{ILS12}. Moreover, by Theorem \ref{def:runningtimeapproximationMTSCSS}, the subgraph $G_{t}=(V,E_{t})$ can be computed in $O(nm)$ time. Since $|E_{t}\setminus B_{s}|<3n$, the number of iteration of the for loop is at most $3n$. Line $7$ can be implemented in linear time using Raghavan's algorithm \cite{SR06}. Therefore, lines $6$--$8$ take $O(nm)$ time.
\end{proof}
\section{$2$-edge-twinless connected graphs}
In this section we show that $2$-edge connected directed graphs have no twinless bridges.
\begin{lemma}
	Let $G=(V,E)$ be a strongly connected graph. The graph $G$ is $2$-edge-connected if and only if $G$ is $2$-edge-twinless-connected.
\end{lemma}
\begin{proof}
	$\Leftarrow$ It follows from definition.	
	$\Rightarrow$ Suppose that $G$ is $2$-edge-connected and it contains a twinless bridge $e=(v,w)$. Then $G \setminus \left\lbrace e \right\rbrace $ is strongly connected. The vertices $v,w$ are in distinct twinless strongly connected components in $G\setminus \left\lbrace e \right\rbrace $. Let $G^{vw}_{tscc}=(V_{tscc}.E_{tscc})$ be the TSCC component graph of the subgraph $G\setminus \left\lbrace e \right\rbrace $. By [\cite{SR06}, Theorem $1$], the underlying graph of $G^{vw}_{tscc}$ is a tree and each edge in this tree corresponds to antiparallel edges in $G\setminus \left\lbrace e \right\rbrace $ . Therefore, there is a simple path $p$ from $w$ to $v$ in $G\setminus \left\lbrace e \right \rbrace $  such that neither $(v,w)$ nor $(w,v)$ belongs to $p$. Let $(i, j)$ be an edge in $p$ that corresponds to an edge in the underlying graph of $G^{vw}_{tscc}$. Therefore, the vertices $i,j$ are not in the same twinless strongly connected component of $G\setminus \left\lbrace e \right\rbrace $. Notice that $G\setminus \left\lbrace (i,j) \right\rbrace $ is not strongly connected, which contradicts that $G$ is $2$-edge-connected.
\end{proof}
\section{Computing $2$-vertex-twinless-connected components}
In this section we present an algorithm for computing the $2$-vertex-twinless-connected components of a directed graph.
The following lemma shows that every two distinct $2$-vertex-twinless-connected components intersect in at most one vertex.
\begin{Lemma}\label{def:2vtcatmostonevertex}
	
	Let $G=(V,E)$ be a twinless strongly connected graph. Let $C^{2vt}_{1},C^{2vt}_{2}$ be distinct $2$-vertex-twinless-connected components. Then $C^{2vt}_{1}$ and $C^{2vt}_{2}$ have at most one vertex in common.
\end{Lemma} 
\begin{proof}
	Suppose for contradiction that $C^{2vt}_{1}$ and $C^{2vt}_{2}$ have at least two vertices in common.
	Let $v \in C^{2vt}_{1} \cup C^{2vt}_{2}$. Then there is a vertex $w\in (C^{2vt}_{1} \cap C^{2vt}_{2})\setminus \left\lbrace  v\right\rbrace  $. For each vertex $x \in C^{2vt}_{1}$, the vertices $x,w$ are in the same twinless strongly connected component of $G\left[ ( C^{2vt}_{1} \cup C^{2vt}_{2} ) \setminus \left\lbrace  v \right\rbrace   \right] $. Moreover, for each vertex  $y$ in $C^{2vt}_{2}$ the vertices $y,w$ lie in the same twinless strongly connected component of $G\left[ ( C^{2vt}_{1} \cup C^{2vt}_{2} ) \setminus \left\lbrace  v \right\rbrace   \right] $. Therefore, by [\cite{SR06}, Lemma $1$] the subgraph $G\left[ ( C^{2vt}_{1} \cup C^{2vt}_{2} ) \setminus \left\lbrace  v \right\rbrace   \right] $ is twinless strongly connected, which contradicts that $C^{2vt}_{1}$ is a $2$-vertex-twinless-connected. 
\end{proof}
\begin{Lemma}\label{def:2vtsubsetof2vc}
	Each $2$-vertex-twinless-connected component in a twinless strongly connected is a subset of $2$-vertex-connected component.	
\end{Lemma}
\begin{proof}
	Immediate from the definition.
\end{proof}
Jaberi proved in his work \cite{Jaberi16} a connection between $2$-vertex-connected components and dominator trees. 
\begin{theorem}\cite{Jaberi16}\label{def:threl2vtcsanddt}
	Let $G=(V,E)$ be a strongly connected graph and let $v$ be an arbitrary vertex in $G$. Let $C^{2vc}$ be a $2$-vertex-connected component of $G$. Then either all elements of $C^{2vc}$ are direct successors of some vertex $w\notin C^{2vc} $ or all elements $C^{2vc}\setminus \lbrace w\rbrace$ are direct successors of some vertex $w \in C^{2vc}$ in the dominator tree $DT(v)$ of the flowgraph $G(v)$.
\end{theorem}

In Figure \ref{fig:2vcbutnot2vtcgexample}, the graph is $2$-vertex-connected but it is not $2$-vertex-twinless-connected. 
\begin{figure}[htbp]
	\centering
	\scalebox{0.96}{
		\begin{tikzpicture}[xscale=2]
		\tikzstyle{every node}=[color=black,draw,circle,minimum size=0.9cm]
		\node (v1) at (-1.2,3.1) {$1$};
		\node (v2) at (-2.5,0) {$2$};
		\node (v5) at (-0.5, -2.5) {$5$};
		\node (v4) at (1,-0.5) {$4$};
		\node (v3) at (1.2,3.1) {$3$};
		\node (v6) at (3.6,0) {$6$};

		\begin{scope}   
		\tikzstyle{every node}=[auto=right]   
		\draw [-triangle 45] (v1) to (v2);
		\draw [-triangle 45] (v2) to [bend left ](v1);
		\draw [-triangle 45] (v1) to (v5);
		\draw [-triangle 45] (v5) to (v3);
		\draw [-triangle 45] (v2) to [bend right ] (v5);
		\draw [-triangle 45] (v5) to (v2);
		\draw [-triangle 45] (v3) to [bend left ](v4);
		\draw [-triangle 45] (v4) to (v5);
		
		\draw [-triangle 45] (v4) to (v3);
		\draw [-triangle 45] (v1) to (v3);
		\draw [-triangle 45] (v2) to (v3);
		\draw [-triangle 45] (v3) to[bend right ] (v1);
		\draw [-triangle 45] (v4) to(v1);
		\draw [-triangle 45] (v2) to (v4);
		
		\draw [-triangle 45] (v6) to [bend left ](v4);
		\draw [-triangle 45] (v4) to (v6);
		\draw [-triangle 45] (v6) to (v3);
		\draw [-triangle 45] (v3) to[bend left ] (v6);
		\draw [-triangle 45] (v5) to[bend right ] (v4);
		\end{scope}
		\end{tikzpicture}}
	\caption{a $2$-vertex-connected graph $G=(V,E)$, but this graph is not $2$-vertex-twinless-connected. This graph contains one $2$-vertex-twinless-connected component $C^{2vt}_{1}=\left\lbrace 1,2,3,4,5\right\rangle $.  Notice that the vertex $4$ is a twinless articulation point in $G$. Furthermore, $C^{2vt}_{1}\setminus \left\lbrace 4 \right\rbrace $ is twinless strongly connected. }
	\label{fig:2vcbutnot2vtcgexample}
\end{figure}
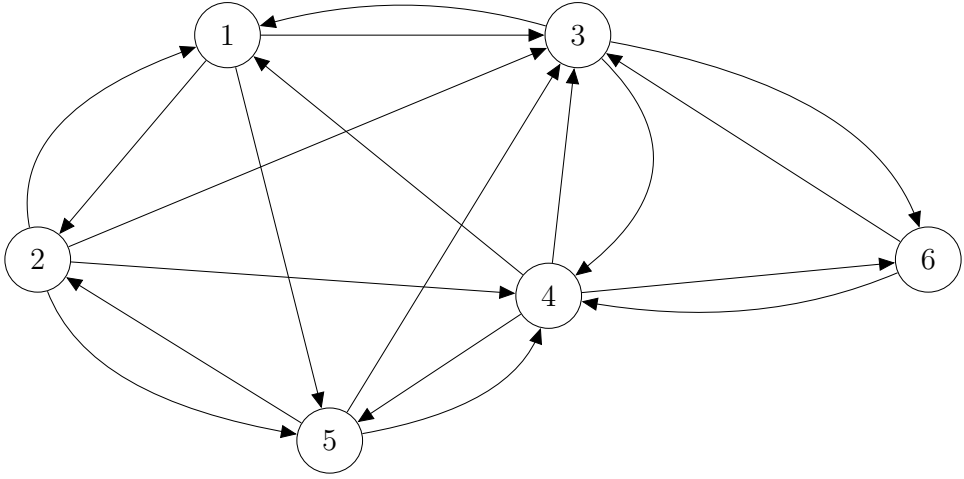
 \begin{Lemma} \label{def:twinlessarticulation}
 Let $G=(V,E)$ be a $2$-vertex-connected graph and let $C^{2vt}$ be a $2$-vertex-twinless-connected component of $G$. Let $v$ be a twinless articulation point in $G$. Then $C^{2vt}\setminus \left\lbrace v \right\rbrace $ is a subset of a twinless strongly connected component of $G\setminus\left\lbrace v \right\rbrace $. 
 \end{Lemma}
\begin{proof}
Notice that $G[C^{2vt}\setminus \left\lbrace v \right\rbrace]$ is twinless strongly connected because it has no twinless articulation points. Therefore, $C^{2vt}\setminus \left\lbrace v \right\rbrace $ is a subset of a twinless strongly connected component of $G\setminus\left\lbrace v \right\rbrace $.
\end{proof}
We use $M(w)$ to denote the set of direct successors of vertex $w$ in the dominator tree of a flowgraph. Algorithm \ref{algo:All2VtsCCsRaedJaberitherel2vccanddt} shows our  algorithm for computing all the $2$-vertex-twinless connected components of a twinless strongly connected graph $G$.
\begin{figure}[htp]
	\begin{myalgorithm}\label{algo:All2VtsCCsRaedJaberitherel2vccanddt}\rm\quad\\[-5ex]
		\begin{tabbing}
			\quad\quad\=\quad\=\quad\=\quad\=\quad\=\quad\=\quad\=\quad\=\quad\=\kill
			\textbf{Input:} A twinless strongly connected graph $G=(V,E)$.\\
			\textbf{Output:} The $2$-vertex-twinless connected components of $G$.\\
			{\small 1}\> \textbf{if} $G$ is $2$-vertex-twinless connected \textbf{then} \\
			{\small 2}\>\> \textbf{Output} $V$. \\
			{\small 3}\> \textbf{else}\\
			{\small 4}\>\> Compute the strong articulation points of $G$.\\
			{\small 5}\>\> \textbf{if} $G$ is not $2$-vertex-connected \textbf{then}\\
			{\small 6}\>\>\> \textbf{if} all vertices in $G$ are strong articulation points \textbf{then}\\
			{\small 7}\>\>\>\> Choose any vertex $v\in V$\\
			{\small 8}\>\>\> \textbf{else}\\
			{\small 9}\>\>\>\> Choose a vertex $v \in V$ that is not a strong articulation .\\
			{\small 10}\>\>\>\> point of $G$.\\
			{\small 11}\>\>\> Compute the dominator trees of the flowgraph $(V,E,v)$ and  \\
			{\small 12}\>\>\> $(V,E^{r},v)$, and choose the dominator tree of  that contains moe \\
			{\small 13}\>\>\>\> non-trivial dominators. \\
			{\small 14}\>\>\> \textbf{for} each vertex $w\in V$ \textbf{do}\\
			{\small 15}\>\>\>\> \textbf{if} $|M(w)| \geq 2$ \textbf{then}  \\
			{\small 16}\>\>\>\>\> \textbf{if} $G[M(w)\cup \lbrace w \rbrace]$ is not twinless strongly connected \textbf{then}\\
			{\small 17}\>\>\>\>\>\>  Compute the twinless strongly connected components of.\\
			{\small 18}\>\>\>\>\>\>  $G[M(w)\cup \lbrace w\rbrace]$.\\
			{\small 19}\>\>\>\>\>\>  \textbf{for} each twinless strongly connected component $C$ of them \textbf{do}\\
			{\small 20}\>\>\>\>\>\>\> \textbf{if} $|C|\geq 3$ \textbf{then} \\
			{\small 21}\>\>\>\>\>\>\>\> Recursively compute the $2$-vertex-twinless connected \\
			{\small 22}\>\>\>\>\>\>\>\>  components of $G[C]$ and \textbf{output} them.  \\
			{\small 23}\>\>\>\>\> \textbf{else}\\
			{\small 24}\>\>\>\>\>\> Recursively compute the $2$-vertex-twinless connected components\\
			 {\small 25}\>\>\>\>\>\>of $G[M(w)\cup \lbrace w \rbrace]$  \\
	    	{\small 26}\>\>\>\>\>\>  and \textbf{output} them.\\
	    	{\small 27}\>\> \textbf{else}\\
	    	{\small 28}\>\>\> find a twinless articulation point $v$ in $G$.\\
	    	{\small 29}\>\>\> \textbf{for} each twinless strongly connected component $C$ of $G\setminus\left\lbrace v \right \rbrace $ \textbf{do}\\
	    	{\small 30}\>\>\>\> \textbf{for} each twinless strongly connected component $C_{c}$ of $G[C\cup\left\lbrace v\right \rbrace ] $ \textbf{do}\\ 
	    	{\small 31}\>\>\>\>\>Recursively compute the $2$-vertex-twinless connected components of $G[C_{c}]$.
	    	
		\end{tabbing}
	\end{myalgorithm}
\end{figure}
\begin{Lemma}
Algorithm \ref{algo:All2VtsCCsRaedJaberitherel2vccanddt} computes the $2$-vertex-twinless connected components of a twinless strongly connected component.
\end{Lemma}
\begin{proof}
	It follows from Theorem \ref{def:threl2vtcsanddt}, Lemma \ref{def:2vtcatmostonevertex}, Lemma \ref{def:2vtsubsetof2vc}, and Lemma \ref{def:twinlessarticulation}.
\end{proof}
\begin{theorem}
	The running time of Algorithm \ref{algo:All2VtsCCsRaedJaberitherel2vccanddt} is $O(n^{2}m)$.
\end{theorem}
\begin{proof}
By [\cite{SR06}, Theorem $2$], Raghavan's algorithm \cite{SR06} is able to test whether a directed graph is twinless strongly connected in linear time. 
 All the strong articulation points of a directed graph can also be calculated in linear time using the algorithm of Italiano et al. \cite{ILS12}. The dominator tree of a flowgraph can be found in linear time \cite{AHLT99,BGKRTW00,GT05,LT79}. For each pair of distinct vertices $x,y$, by Lemma \ref{def:2vtcatmostonevertex}, the subgraph $|M(x)\cup\left\lbrace x \right\rbrace |,|M(y)\cup\left\lbrace y \right\rbrace |$ are disjoint.
  Raghavan's algorithm \cite{SR06} can find all the twinless strongly connected components of a directed graph in linear time. Moreover, the edge sets of the subgraphs considered in Lines $30$--$31$ are disjoint. The recursion depth is at most $n$ because the graph in a recursion call has less vertices than the original one.  
\end{proof}

\section{Open problems}
We leave as an open problem  whether there is a linear time algorithm for testing $2$-vertex-twinless connectivity.
Another open question is  whether the twinless articulation points, twinless bridges, and the $2$-vertex-twinless connected components of a directed graph can be calculated in linear time.

\end{document}